\documentclass[11pt,reqno]{amsart}

\usepackage{amscd,amssymb,amsmath,amsthm}
\usepackage[russian]{babel}
\usepackage[arrow,matrix]{xy}
\usepackage{graphicx}
\usepackage{color}
\usepackage{cite}
\newtheorem{thm}{Теорема}
\newtheorem{defn}{Определение}
\newtheorem{lemma}{Лемма}
\newtheorem{pro}{Утверждение}
\newtheorem{rk}{Замечание}
\newtheorem{cor}{Следствие}

\numberwithin{equation}{section} 

\newcommand{\bea}{\begin{eqnarray}}
\newcommand{\eea}{\end{eqnarray}}

\newcommand{\Q}{\mathbb{Q}}

\begin{document}
\title[$p$-адические квазимеры Гиббса для модели Ваннименуса]{$p$-адические квазимеры Гиббса для модели Ваннименуса на дереве Кэли}

\author{О.\ Н.\ Хакимов}

 \address{О.\ Н.\ Хакимов\\ Институт математики, ул. Дурмон йули, 29, Ташкент, 100125,
Узбекистан.} \email {hakimovo@mail.ru}

\begin{abstract} В этой работе мы изучим $p$-адические
квазимеры Гиббса для модели Ваннименуса на дереве Кэли порядка
два. Изучена ограниченность для трансляционно-инвариантных
$p$-адических квазимер Гиббса. Также будут исследованы
периодические $p$-адические квазимеры Гиббса.
\end{abstract}
\maketitle

{\bf{Ключевые слова:}} дерево Кэли, конфигурация, квазимера
Гиббса, модель Ваннименуса, трансляционно-инвариантная мера,
$p$-адические числа.

\section{Введение}

Описание предельных мер Гиббса для данного гамильтониана является
одним из основных задач в теории гиббсовских мер. Полный анализ
множества таких мер является довольно трудоемким. По этой причине
большая часть работ по этой тематике посвящены изучению
гиббсовских мер на дереве Кэли \cite{B,17}.

Известно \cite{24,34,48}, что $p$-адические модели в физике не
могут быть описаны, используя обычную теорию вероятностей. В
\cite{24} абстрактная $p$-адическая теория вероятностей была
развита посредством теории неархимедовых мер.
Вероятностные процессы на поле $p$-адических чисел были изучены
многими авторами (см. \cite{1a,51,pImpq,rh,MNGM1,MNGM,bookroz}). Не
архимедовый аналог теоремы Колмогорова был доказан в \cite{16}.

В работе \cite{pImpq} были изучены $p$-адические меры Гиббса для
модели Изинга с четырмя конкурирующими взаимодействиями на дереве
Кэли. Доказаны, что множество $p$-адических мер Гиббса состоит из
единственной трансляционно-инвариантной меры Гиббса. Более того,
эта мера является ограниченной. В работах \cite{MNGM1,MNGM} были
изучены трансляционно-инвариантные $p$-адические квазимеры Гиббса
для модели Поттса на дереве Кэли порядка два. Показаны, что
множество таких мер может состоять более из одного элемента. А в
работе \cite{pImpq} также изучены трансляционно-инвариантные
$p$-адические меры Гиббса для модели Ваннименуса на дереве Кэли.
Было доказано, что если $J<0$, то существуют шесть
трансляционно-инвариантные $p$-адические квазимеры Гиббса.

Настоящую работу можно считать как продолжение работы
\cite{pImpq}. В работе будем изучать проблемы ограниченности
трансляционно-инвариантных $p$-адических квазимер Гиббса для
модели Ваннименуса. Также будем исследовать периодические
$p$-адические квазимеры Гиббса.
\section{Определения и факты}

\subsection{$p$-адические числа и меры.} Каждое рациональное число $x\neq 0$ может быть представлено в виде
$x=p^r\frac{n}{m}$, где $r,n\in \mathbb{Z}, m$-- положительное
число, $(n,m)=1$, причем $m$ и $n$ не делятся на $p$ и $p$ --
фиксированное простое число. $p$-Адическая норма $|x|_p$
определяется по формуле
$$
|x|_p= \left\{\begin{array}{ll}
p^{-r},& \text{ если }  x\neq 0,\\
0,& \text{ если } x=0.
 \end{array}\right.
$$

Эта норма удовлетворяет сильному неравенству треугольника:
$$
|x+y|_p\leq\max\{|x|_p,|y|_p\}.
$$
Это свойство показывает  неархимедовость нормы.

Из этого свойства непосредственно следуют следующие утверждения:

1) если  $|x|_p\neq |y|_p$, то $|x-y|_p=\max\{|x|_p,|y|_p\}$;

2) если  $|x|_p=|y|_p$, то  $|x-y|_p\leq |x|_p$;

Пополнение поля рациональных чисел $\mathbb{Q}$ по  $p$-адической
норме приводит к полю $p$-адических чисел $\mathbb{Q}_p$ для
каждого простого $p$ (см. \cite{29}).

Начиная с поля рациональных чисел $\mathbb{Q}$, мы можем получить
либо поле вещественных чисел $\mathbb{R}$, либо одно из полей
$p$-адических чисел $\mathbb{Q}_p$ (теорема Островского).

Каждое $p$-адическое число $x\neq0$ имеет единственное
каноническое разложение
\begin{equation}\label{ek}
x = p^{\gamma(x)}(x_0+x_1p+x_2p^2+\dots),
\end{equation}
где $\gamma=\gamma(x)\in \mathbb Z$ и $x_j$ целые числа, $0\leq
x_j \leq p - 1$, $x_0 > 0$, $j = 0, 1,2, ...$ (см
\cite{29,41,48}). В этом случае $|x|_p = p^{-\gamma(x)}$.

\begin{thm}\label{tx2}\cite{48} Уравнение
$x^2 = a$, $0\ne a =p^{\gamma(a)}(a_0 + a_1p + ...), 0\leq a_j
\leq p - 1$, $a_0 > 0$ имеет решение $x\in \Q_p$ тогда и только
тогда, когда выполняются следующие:

i) $\gamma(a)$ четное;

ii) $y^2=a_0(\operatorname{mod} p)$ разрешимо, если $p\ne 2$;
$a_1=a_2=0$, если $p=2$.
\end{thm}
\begin{cor}\label{sqrt-1} \cite{48}
Для того чтобы уравнение $x^2 =-1$ имело решение в $\mathbb Q_p$,
необходимо и достаточно, чтобы $p\equiv1(\operatorname{mod} 4)$.
\end{cor}
Для $a\in \mathbb{Q}_p$ и $r> 0$ обозначим
$$B(a, r) = \{x\in \mathbb{Q}_p : |x-a|_p < r\}.$$

$p$-адический {\it логарифм} определяется как ряд
$$\log_p(x) =\log_p(1 + (x-1)) =
\sum_{n=1}^\infty (-1)^{n+1}{(x-1)^n\over n},$$ который сходится
для  $x\in B(1, 1)$;  $p$-адическая экспонента определяется как
$$\exp_p(x) =\sum^\infty_{n=0}{x^n\over n!},$$
которая сходится для $x \in B(0, p^{-1/(p-1)})$.
\begin{lemma}\label{el} Пусть $x\in B(0, p^{-1/(p-1)})$. Тогда
$$|\exp_p(x)|_p = 1,\ \ |\exp_p(x)-1|_p = |x|_p, \ \ |\log_p(1 + x)|_p = |x|_p,$$
$$\log_p(\exp_p(x)) = x,\ \ \exp_p(\log_p(1 + x)) = 1 + x.$$
\end{lemma}

Более подробно об основах $p$-адического анализа и $p$-адической
математической физики можно найти в \cite{29,41,48}.

Пусть $(X,\mathcal B)$ измеримое пространство, где $\mathcal B$
алгебра подмножеств в $X$. Функция $\mu:\mathcal B\to\mathbb Q_p$
называется $p$-адической мерой, если для любого набора
$A_1,...,A_n\in\mathcal B$ такого, что $A_i\cap A_j=\varnothing,\
i\neq j$ имеет место
$$\mu\bigg(\bigcup_{j=1}^nA_j\bigg)=\sum_{j=1}^n\mu(A_j).$$
$p$-Адическая мера называется вероятностной, если $\mu(X)=1$ (см.
\cite{16}).

\subsection{Дерево Кэли}

Дерево Кэли $\Gamma^k=(V,L)$ порядка $k\geq1$ есть бесконечное
дерево (граф без циклов), из каждой вершины которого выходит ровно
$k+1$ ребер, $V\ -$ множество вершин и $L\ -$ множество ребер. Две
вершины $x$ и $y$ называются {\it ближайшими соседями}, если
существует ребро $l\in L$ соединяющий их и пишется как $l=\langle
x,y\rangle$. Расстояние $d(x,y)\ -$ число ребер кратчайшей пути,
соединяюшей $x$ и $y$.

Пусть $x^0\in V$ фиксированная точка. Введем обозначения:
$$W_n=\{x\in V|d(x,x^0)=n\},\qquad V_n=\bigcup_{m=0}^nW_m,$$

и
$$S(x)=\{y\in W_{n+1}:d(x,y)=1\},\quad x\in W_n.$$
Обычно говорят, что $S(x)$ это множество прямых потомков элемента
$x$. Две вершины $y$ и $z$ называются {\it следующими ближайшими
соседями}, если существует вершина $x\in V$ такая, что $y,z\in
S(x)$ и обозначается через $\rangle y,z\langle$.
\subsection{Модель Ваннименуса}

Мы рассмотрим $p$-адическую модель Ваннименуса на дереве Кэли
порядка два.

Пусть $\mathbb Q_p$ поле $p$-адических чисел и $\Phi=\{-1;1\}$.
Конфигурация $\sigma$ в $V$ определяется как функция $x\in
V\to\sigma(x)\in\Phi$; аналогично определяются конфигурации
$\sigma_n$ и $\sigma^{(n)}$ на $V_n$ и $W_n$, соответственно.
Множество всех конфигураций на $V$ (соответственно $V_n,\ W_n$)
обозначается через $\Omega=\Phi^V$ (соответственно
$\Omega_{V_n}=\Phi^{V_n},\ \Omega_{W_n}=\Phi^{W_n}$). Для
конфигураций $\sigma_{n-1}\in\Omega_{V_n}$ и
$\varphi^{(n)}\in\Omega_{W_n}$ определим

\[  (\sigma_{n-1}\vee\varphi^{(n)})(x)=\left\{\begin{array}{ll}
\sigma_{n-1}(x),& \text{если}\  x\in V_{n-1},\\
\varphi^{(n)}(x),& \text{если}\  x\in W_n.
 \end{array}\right.\]

Очевидно, что $\sigma_{n-1}\vee\varphi^{(n)}\in\Omega_{V_n}.$\\
Гамильтониан $H_n:\Omega_{V_n}\to\mathbb Q_p$ $p$-адической модели
Ваннименуса имеет следующий вид
\begin{equation}\label{Ham}H_n(\sigma)=J_1\sum_{\langle x,y\rangle\in L_n}\sigma(x)\sigma(y)+J_2\sum_{\rangle x,y\langle\atop{x,y\in
V_n}}\sigma(x)\sigma(y).
\end{equation}
где $J_1,J_2\in\mathbb Q_p$.
\begin{rk}
Заметим, что модель Ваннименуса является обобщением модели Изинга.
Если в модели Ваннименуса $J_2=0$, то получается модель Изинга.
Более подробно о модели Ваннименуса можно найти в книге
\cite{bookroz}.
\end{rk}
\subsection{Построение $p$-адической квази меры Гиббса.}
Следуя работ \cite{MNGM1,MNGM} построим $p$-адическую меру Гиббса
для модели (\ref{Ham}). Как и в классическом случае, мы рассмотрим
специальный класс меры Гиббса.

Пусть $h:x\to h_x\in\mathbb Q_p$ $p$-адическая функция на $V$.
Рассмотрим $p$-адическое вероятностное распределение $\mu_h^{(n)}$
на $\Omega_{V_n}$, которое определяется как
\begin{equation}\label{mu}
\mu_h^{(n)}(\sigma_n)=Z_{n,h}^{-1}p^{H_n(\sigma_n)}\prod_{x\in
W_n}h_x^{\sigma(x)},\qquad n=1,2,...,
\end{equation}
где $Z_{n,h}$ нормирующая константа
\begin{equation}\label{z}
Z_{n,h}=\sum_{\varphi\in\Omega_{V_n}}p^{H_n(\varphi)}\prod_{x\in
W_n}h_x^{\varphi(x)}.
\end{equation}\\
Говорят, что $p$-адическое вероятностное распределение $\mu_h^{(n)}$
согласовано, если $\mbox{для всех}\ n\geq1$ и
$\sigma_{n-1}\in\Omega_{V_{n-1}},$

\begin{equation}\label{us}
\sum_{\varphi\in\Omega_{W_n}}\mu_h^{(n)}(\sigma_{n-1}\vee\varphi){\bf1}
(\sigma_{n-1}\vee\varphi\in\Omega_{V_n})=\mu_h^{(n-1)}(\sigma_{n-1}).
\end{equation}
В этом случае по теореме Колмогорова \cite{16} существует
единственная мера $\mu_h$ на $\Omega$ такая, что
$\mu_h(\{\sigma\big|_{V_n}=\sigma_n\})=\mu_h^{(n)}(\sigma_n)$ для
всех $n\in\mathbb N$ и $\sigma_n\in\Omega_{V_n}$.

\begin{defn}
$p$-адическая вероятностная мера $\mu$ называется $p$-адической
квазимерой Гиббса, если существует $p$-адическая функция $h$ от
$x\in V$ такая, что
$$\mu(\sigma\in\Omega:\sigma|_{V_n}=\sigma_n)=\mu_h^{(n)}(\sigma_n),\qquad \mbox{при всех}\ \sigma_n\in\Omega_{V_n},\qquad n\in\mathbb N.$$
Здесь $\mu_h^{(n)}$ определена как (\ref{mu}),(\ref{z}).
\end{defn}
Обозначим через $\mathcal {QG}(H)$  множество всех $p$-адических
квазимер Гиббса, соответсвующих функциям $h=\{h_x,\ x\in V\}$.
Рассмотрим гамильтониан (\ref{Ham}) в случае $J=J_1=J_2\in\mathbb
Z$.
\begin{rk} Заметим, что меры $\mu_h$ и $\mu_{-h}$ соответствующие функциями $h$ и $-h$ одинаковы.
\end{rk}
\begin{pro}\cite{pImpq}
$p$-адическая вероятностная мера $\mu_h^{(n)},\ n=1,2,...$ удовлетворяет условию согласованности (\ref{us})
тогда и только тогда, когда для любого $x\in V$ имеет место следующее:
\begin{equation}\label{qequcc}
u_x=\frac{\theta^2u_yu_z+u_y+u_z+1}{u_yu_z+u_y+u_z+\theta^2},
\end{equation}
здесь $\theta=p^{2J},\ u_x=h_x^2$ и $S(x)=\{y,z\}$.
\end{pro}
\begin{rk}
Известно, что вещественнозначные меры Гиббса возникают во многих
проблемах теории вероятностей и статистической механики. Эта мера
определяется с помощью функции \,"экпоненты\,". Аналогично
$p$-адическая мера Гиббса определяется с помощью $p$-адической
\,"экпоненты\,"\ $\exp_p(x)$. Но область определения и область
значения функции $\exp_p(x)$ не очень хороша для работы над ними.
Поэтому для многих моделей, в частности для модели Изинга
существует только одна $p$-адическая мера Гиббса. Для того, чтобы
получить широкий класс $p$-адических мер Гиббса в работе
\cite{MNGM1} были введены понятие $p$-адической квазимеры Гиббса,
которая определяется с помощью функции $p^x$. В работах
\cite{MNGM1,MNGM} для модели Поттса и в работе \cite{pImpq} для
модели Ваннименуса показаны, что множество $\mathcal {QG}(H)$
шире, чем множество всех $p$-адических мер Гиббса. Более того,
$p$-адические квазимеры Гиббса могут быть неограниченными (см.
\cite{MNGM}).
\end{rk}
\section{Трансляционно-инвариантная квази мера Гиббса}

Решения уравнения (\ref{qequcc}) вида $u_x=u\in\mathbb Q_p,\ x\neq
x_0$ называются {\it трансляционно-инвариантными}. Соответствующая
$p$-адическая квазимера Гиббса
называется трансляционно-инвариантной мерой Гиббса.\\
Подставляя $u$ вместо $u_x$ для всех $x\neq x_0$, из уравнения
(\ref{qequcc}) получим
\begin{equation}\label{qtri}
u=\frac{\theta^2u^2+2u+1}{u^2+2u+\theta^2}.
\end{equation}
Легко проверить, что $u_0=1$ является решением уравнение
(\ref{qtri}). Так как уравнение (\ref{qtri}) можно рассмотреть как
кубическое уравнение, то для других решений (если они существуют)
имеем формальную запись
\begin{equation}\label{sol}
u_{1,2}=\frac{\theta^2-3\pm\sqrt{(1-\theta^2)(5-\theta^2)}}{2}.
\end{equation}
Из \cite{pImpq} известны следующие теоремы:
\begin{thm}\label{ferti}
Пусть $J>0$. Тогда верны следующие:

(i) Если $p\in\{2,3,5\}$ то существует единственная
трансляционно-инвариантная $p$-адическая квазимера Гиббса
$\mu_{h_0}$;

(ii) Пусть $p>5$ и $x_0$ является решением сравнения $x^2\equiv 5
\,(\operatorname{mod} p)$. Если сравнение $x^2+6\equiv 2x_0
\,(\operatorname{mod} p)$ разрешимо, то существуют три
трансляционно-инвариантные
$p$-адические квазимеры Гиббса: $\mu_{h_0},\ \mu_{h_1},\ \mu_{h_2}$.\\
Здесь $h_0=1,\ h_1=\sqrt{u_1},\ h_2=\sqrt{u_2}$.
\end{thm}
\begin{thm}\label{antiferti}
Пусть $J<0$. Тогда существуют три трансляционно-инвариантных
$p$-адических квазимер Гиббса $\mu_{h_0},\ \mu_{h_1},\ \mu_{h_2}$.\\
\end{thm}

\subsection{Ограниченность трансляционно-инвариантных $p$-адических квазимер Гиббса}

\begin{lemma}\label{zrec}
Пусть $h$ является решением уравнения (\ref{qequcc}) и $\mu_h$
соответствующая $p$-адическая квазимера Гиббса. Тогда для
нормирующей константы $Z_{n,h}$ (см. (\ref{z})) имеет место
равенство
\begin{equation}
Z_{n+1,h}=A_{n,h}Z_{n,h},
\end{equation}
где $A_{n,h}$ определяется по формуле (\ref{Anh}).
\end{lemma}
\begin{proof}
Так как $h$ является решением уравнения (\ref{qequcc}), то для любого $x\in V$ существует константа $a_h(x)\in\mathbb Q_p$ такая, что
\begin{equation}\label{a(x)}
\sum_{\varphi\in\Omega_{W_{n+1}}}p^{J(\sigma(x)(\varphi(y)+\varphi(z))+\varphi(y)\varphi(z))}h_y^{\varphi(y)}h_z^{\varphi(z)}=
a_h(x)h_x^{\sigma(x)},
\end{equation}
здесь $S(x)=\{y,z\}$ и $\sigma\in\Omega_{V_n}$.\\
Отсюда
\begin{equation}\label{a(x)(x)}
\prod_{x\in W_{n}}\sum_{\varphi\in\Omega_{W_{n+1}}}p^{J(\sigma(x)(\varphi(y)+\varphi(z))+\varphi(y)\varphi(z))}h_y^{\varphi(y)}h_z^{\varphi(z)}
=\prod_{x\in W_{n}}a_h(x)h_x^{\sigma(x)}=A_{n,h}\prod_{x\in W_n}h_x^{\sigma(x)},
\end{equation}
где
\begin{equation}\label{Anh}
A_{n,h}=\prod_{x\in W_n}a_h(x).
\end{equation}
Из (\ref{mu}) и (\ref{a(x)(x)}) получим
$$\sum_{\sigma\in\Omega_{V_n}}\sum_{\varphi\in\Omega_{W_{n+1}}}\mu_h^{(n+1)}(\sigma\vee\varphi)=
\sum_{\sigma\in\Omega_{V_n}}\sum_{\varphi\in\Omega_{W_{n+1}}}\frac{1}{Z_{n+1,h}}p^{H(\sigma\vee\varphi)}\prod_{x\in W_{n+1}}h_x^{\varphi(x)}$$
$$
=\frac{A_{n,h}}{Z_{n+1,h}}\sum_{\sigma\in\Omega_{V_n}}p^{H(\sigma)}\prod_{x\in W_n}h_x^{\sigma(x)}=\frac{A_{n,h}}{Z_{n+1,h}}Z_{n,h}=1.
$$
\end{proof}
Пусть $h$ является решением уравнения (\ref{qequcc}). Для $h$
найдем $a_h(x)$. Фиксируем точку $x\in V$ и перепишем (\ref{a(x)})
для случаев $\sigma(x)=1$ и $\sigma(x)=-1$. При $\sigma(x)=1$ и
$\sigma(x)=-1$ соответственно имеем
$$p^{3J}h_yh_z+p^{-J}h_y^{-1}h_z+p^{-J}h_yh_z^{-1}+p^{-J}h_y^{-1}h_z^{-1}=a(x)h_x$$

$$p^{-J}h_yh_z+p^{-J}h_y^{-1}h_z+p^{-J}h_yh_z^{-1}+p^{3J}h_y^{-1}h_z^{-1}=a(x)h_x^{-1}.$$
Умножая эти равенства, получим
\begin{equation}\label{a2(x)}
a_h(x)=\frac{\left(\left(p^{4J}h_y^2h_z^2+h_y^2+h_z^2+1)(h_y^2h_z^2+h_y^2+h_z^2+p^{4J}\right)\right)^\frac{1}{2}}{p^Jh_yh_z}.
\end{equation}
Для трансляционно-инвариантных решений $h$ формула (\ref{a2(x)}) имеет вид
\begin{equation}\label{ah}
a_h=\frac{\left(\left(p^{4J}h^4+2h^2+1)(h^4+2h^2+p^{4J}\right)\right)^\frac{1}{2}}{p^Jh^2}.
\end{equation}
\subsubsection{Cлучай $J>0$.}

\begin{lemma}\label{minH}
Для любой конфигурации $\sigma\in\Omega_{V_n}$ и $n\geq1$ имеет
место
$$\left|p^{H_n(\sigma)}\right|_p\leq p^{J(2^n-1)}.$$
\end{lemma}
\begin{proof} Легко убедиться, что $H_n(\sigma)\geq -J(2^n-1)$. Заметим, что Гамильтониан достигает своего минимума.
Например, конфигурация $\sigma\in\Omega_{V_n}$ определенная как
$$\sigma(y)\sigma(z)=-1,\ \mbox{при всех}\ x\in V_{n-1},\ S(x)=\{y,z\}$$
дает минимальное значение гамильтониана.
\end{proof}

\begin{lemma}\label{normh} $\left|h_0\right|_p=\left|h_1\right|_p=\left|h_2\right|_p=1$.
\end{lemma}

\begin{proof}
Очевидно, что $|h_0|_p=1$, так как $h_0=1$. В силу теоремы
\ref{ferti} решения $h_1,\ h_2$ могут существовать лишь только при
$p>5$. Более того, в силу свойства 1) пункта 2.1 имеем
$$|h_1|_p=\left|\sqrt{\frac{p^{4J}-3+\sqrt{p^{8J}-6p^{4J}+5}}{2}}\right|_p=\left|\sqrt{2\sqrt{5}-6}\right|_p=1.$$
Аналогично проверяется $|h_2|_p=1$.
\end{proof}

\begin{lemma}\label{normZ} Для нормирующей константы $Z_{n,h_i},\ i=0,1,2$ верны следующие:\\
i) $|Z_{n,h_1}|_p=|Z_{n,h_2}|_p=p^{J(2^n-2)}$;\\
ii) $|Z_{n,h_0}|_p=\left\{\begin{array}{ll}
p^{J(2^n-2)},& \text{ если }  p\neq3,\\
p^{(J-1)(2^n-2)},& \text{ если } p=3.
\end{array}\right.$
\end{lemma}
\begin{proof} {\it i}) Из (\ref{ah}) для $h_1$ имеем
$$|a_{h_1}|_p=\left|\frac{\left(\left(p^{4J}h_1^4+2h_1^2+1)(h_1^4+2h_1^2+p^{4J}\right)\right)^\frac{1}{2}}{p^Jh_1^2}\right|_p
=$$
$$\left|p^{-J}\sqrt{(2\sqrt{5}-4)(\sqrt{5}+1)}\right|_p=\left|p^{-J}\sqrt{6-2\sqrt{5}}\right|_p=p^J$$
Далее, так как
$Z_{n,h}=a_h^{|V_{n-1}|}$ и $|V_{n-1}|=2^n-2$, то
$$|Z_{n,h_1}|_p=p^{J(2^n-2)}.$$
Аналогично проверяется $|Z_{n,h_2}|_p=p^{J(2^n-2)}$.

{\it ii}) Так как $h_0=1$, то из (\ref{ah}) получим
$$|a_{h_0}|_p=\left|\frac{\left(\left(p^{4J}+3)(3+p^{4J}\right)\right)^\frac{1}{2}}{p^J}\right|_p
=\left|3p^{-J}\right|_p=\left\{\begin{array}{ll}
p^J,& \text{ если }  p\neq3,\\
p^{J-1},& \text{ если } p=3.
 \end{array}\right.$$
Отсюда,
$$|Z_{n,h_0}|_p=\left\{\begin{array}{ll}
p^{J(2^n-2)},& \text{ если }  p\neq3,\\
p^{(J-1)(2^n-2)},& \text{ если } p=3.
\end{array}\right.$$
\end{proof}
\begin{thm}
i) Если $p\neq3$, то все трансляционно-инвариантные $p$-адические квазимеры Гиббса являются ограниченными.\\
ii) Если $p=3$, то существует единственная трансляционно-инвариантная
$p$-адическая квазимера Гиббса $\mu_{h_0}$. Причем
она является неограниченной.
\end{thm}
\begin{proof}{\it i}) Пусть $p\neq3$. В этом случае в силу леммы \ref{normZ} мы имеем $|Z_{n,h_i}|_p=p^{J(2^n-2)},\ i=0,1,2$. В силу лемм \ref{minH},\ref{normh} для любой конфигурации $\sigma\in\Omega_{V_n}$ и $n=1,2,...$ имеем
$$\left|\mu_{h_i}^{(n)}(\sigma)\right|_p=\left|\frac{p^{H_n(\sigma)}\prod_{x\in W_n}h_i^{\sigma(x)}}{Z_{n,h_i}}\right|_p\leq\frac{p^{J(2^n-2)}}{p^{J(2^n-2)}}=1,\qquad i=0,1,2.$$
Это означает, что в этом случае все трансляционно-инвариантные $p$-адические квазимеры Гиббса $\mu_{h_i},\ i=0,1,2$ ограничены.\\

{\it ii}) Пусть $p=3$. В этом случае в силу теоремы \ref{ferti}
существует единственная трансляционно-инвариантная $p$-адическая квазимера
Гиббса $\mu_{h_0}$. Покажем, что она неограничена.
Определим конфигурацию $\sigma$ следующим образом
$$\sigma(y)\sigma(z)=-1,\ \mbox{при всех}\ x\in V_{n-1},\ S(x)=\{y,z\}$$
Тогда в силу лемм \ref{normh},\ref{normZ} для нормы меру $\mu_{h_0}$ в этой конфигурации имеем
$$\left|\mu_{h_0}^{(n)}(\sigma)\right|_p=\left|\frac{p^{H_n(\sigma)}\prod_{x\in W_n}h_0}{Z_{n,h_0}}\right|_p=\frac{p^{J(2^n-2)}}{p^{(J-1)(2^n-2)}}=p^{2^n-2}.$$
Отсюда получим $$\left|\mu_{h_0}^{(n)}(\sigma)\right|_p\to\infty\qquad \mbox{при}\ n\to\infty.$$
\end{proof}

\subsubsection{Случай $J<0$.}

\begin{lemma}\label{maxH}
$\left|p^{H_n(\sigma)}\right|_p\leq p^{-J(3\cdot2^n-5)}$ для любой конфигурации $\sigma\in\Omega_{V_n}$ и $n\geq1$.
\end{lemma}
\begin{proof} Заметим, что Гамильтониан достигает своего минимума при конфигурации $\sigma\in\Omega_{V_n}$, которая принимает значение 1 при всех $x\in V_n$.
\end{proof}

\begin{lemma}\label{normh2}
$|h_0|_p=1,\qquad |h_1|_p=p^{-2J},\qquad |h_2|_p=p^{2J}$.
\end{lemma}

\begin{proof} Очевидно, что $|h_0|_p=1$. Для нормы $h_1$ имеем
$$|h_1|_p=\left|p^{2J}\sqrt{\frac{1-3p^{-4J}+\sqrt{1-6p^{-4J}+5{p^{-4J}}}}{2}}\right|_p=p^{-2J}.$$
Так как $h_i=\sqrt{u_i},\ i=1,2$ и $u_1\cdot u_2=1$, то
$|h_2|_p=p^{2J}$.
\end{proof}

\begin{lemma}\label{normZ2} Для нормирующей константы $Z_{n,h_i},\ i=0,1,2$ верны следующие
$$|Z_{n,h_i}|_p=p^{-J(5\cdot2^n-10)},\ i=1,2\qquad |Z_{n,h_0}|_p=p^{-J(3\cdot2^n-6)}.$$
\end{lemma}
\begin{proof} В силу леммы \ref{normh2} имеем $h_1=p^{2J}\varepsilon$ где $|\varepsilon|_p=1$. Следовательно,
$$|a_{h_1}|_p=\left|\frac{\left(\left(p^{12J}\varepsilon^4+2p^{4J}\varepsilon^2+1)
(p^{8J}\varepsilon^4+2p^{4J}\varepsilon^2+p^{4J}\right)\right)^\frac{1}{2}}{p^{5J}\varepsilon^2}\right|_p
=p^{-5J}.$$
Отсюда,
$$|Z_{n,h_1}|_p=p^{-J(5\cdot2^n-10)}.$$
Аналогично проверяются $|Z_{n,h_2}|_p=p^{-J(5\cdot2^n-10)}$ и $|Z_{n,h_0}|_p=p^{-J(3\cdot2^n-6)}$.
\end{proof}
\begin{thm}
Все трансляционно-инвариантные $p$-адические квазимеры Гиббса
ограничены.
\end{thm}
Доказательство следует из лемм \ref{maxH},\ref{normh2},\ref{normZ2}.
\section{Периодическая $p$-адическая квазимера Гиббса}
Будем исследовать следующее уравнение:
\begin{equation}\label{per}
u=f(f(u)),\qquad\mbox{где}\ f(u)=\frac{\theta^2u^2+2u+1}{u^2+2u+\theta^2}
\end{equation}
Заметим, что множество решений уравнения (\ref{per}) содержит
решения уравнения $u=f(u)$. Но нас интересует только периодические
(не являющиеся трансляционно-инвариантными) меры. Поэтому
рассмотрим уравнение
$$\frac{f(f(u))-u}{f(u)-u}=0,$$
из которого получим:
\begin{equation}\label{per1}
\theta^2u^2+(\theta^2+1)u+\theta^2=0.
\end{equation}
Если существует $\sqrt{1+2\theta^2-3\theta^4}$ в $\mathbb Q_p$, то
\begin{equation}\label{solper}
u_{3,4}=\frac{-1-\theta^2\pm\sqrt{1+2\theta^2-3\theta^4}}{2\theta^2}.
\end{equation}
являются решениями уравнения (\ref{per1}). Обозначим
$D(\theta)=1+2\theta^2-3\theta^4$. Сначало мы должны проверить
существование $\sqrt{D(\theta)}$ в $\mathbb Q_p$. Затем изучим
существование чисел $\sqrt{u_3}$ и $\sqrt{u_4}$. Заметим, что из
существования одного из них получаем существование второго.
Действительно, предположим, что $\sqrt{u_3}$ существует в $\mathbb
Q_p$. Тогда мы имеем
\begin{equation}\label{u_3u_4}
u_3\cdot u_4=\frac{(1+\theta^2)^2-(1+2\theta^2-3\theta^4)}{4\theta^4}=1.
\end{equation}
Так как $\sqrt{u_3}\in\mathbb Q_p$, то из (\ref{u_3u_4}) получим $\sqrt{u_4}\in\mathbb Q_p$.

\begin{rk}\label{perrk}
Так как существование одного из чисел $\sqrt{u_3}$ и $\sqrt{u_4}$
влечет за собой существование другого, то мы заключаем, что либо
не существует 2-периодическая $p$-адическая квазимеры Гиббса, либо
существуют две 2-периодические $p$-адические квазимер Гиббса.
\end{rk}
Обозначим через $\mu^{per}_1$ (соотв. $\mu^{per}_2$) $p$-адическая квази мера Гиббса соответствующей вектору $(h_3,h_4)$ (соотв. $(h_4,h_3)$).
\subsection{Случай $J>0$} В этом случае в силу Теоремы \ref{tx2}
существует $\sqrt{D(\theta)}$ для любого простого числа $p$.
Теперь проверим существование $\sqrt{u_3}$ в $\mathbb Q_p$.\\
Пусть $p=2$. Тогда имеем
$$u_3=\frac{-1-2^{4J}+\sqrt{1+2^{4J+1}-3\cdot 2^{8J}}}{2^{4J+1}}=\frac{-1-2^{4J}+1+2+2^2+\dots}{2^{4J+1}}=2^{-4J}(1+2+\dots)$$
В силу Теоремы \ref{tx2} следует, что $\sqrt{u_3}$ не существует в $\mathbb Q_p$.\\
Пусть $p\neq2$. Тогда имеем
$$u_4=\frac{-1-p^{4J}-\sqrt{1+2p^{4J}-3p^{8J}}}{2p^{4J}}=\frac{-1- p^{4J}-1-p^{4J}-\dots}{2p^{4J}}=\frac{-1+a_1p+a_2p^{2}+\dots}{p^{4J}}.$$
Отсюда видно, что существование $\sqrt{u_3}$ эквивалентно существованию $\sqrt{-1}$. В силу Следствие \ref{sqrt-1} число $\sqrt{-1}$ существует в $\mathbb Q_p$ тогда и только тогда, когда $p\equiv1(\operatorname{mod} 4)$. Таким образом мы получили
\begin{thm}\label{perfer}
Если $p\equiv1(\operatorname{mod} 4)$, то для модели (\ref{Ham}) существуют две 2-периодических $p$-адических квазимер Гиббса:
$\mu^{per}_1$ и $\mu^{per}_2$.
\end{thm}

\subsection{Cлучай $J<0$} В этом случае $|\theta|_p>1$. Тогда из $D(\theta)=\theta^4(-3+2\theta^{-2}+\theta^{-4})$ видно, что существования
$\sqrt{D(\theta)}$ и $\sqrt{-3}$ эквивалентны. В таблице 1 при
маленких простых чисел $p$ показаны условия, при которых
существует $\sqrt{D(\theta)}$
\begin{center}
\begin{tabular}{|l|c|c|c|c|c|c|c|c|c|r|} \hline
$p$ & $2$ & $3$ & $5$ & $7$ & $11$ & $13$ & $17$ & $19$\\
\hline
$\sqrt{D(\theta)}$ & $-$ & $-$ & $-$ & $+$ & $-$ & $+$ & $-$ & $-$\\
\hline
\end{tabular}
\end{center}
\begin{center}{Таблица 1.}
\end{center}
\begin{thm}\label{perant}
i) Если $p\in\{2,3\}$, то не существует периодическая $p$-адическая квазимера Гиббса.\\
ii) Пусть $p>3$. Если сравнение $x^2+3\equiv0
\,(\operatorname{mod} p)$ не разрешимо в $\mathbb Q_p$, то не существует периодическая $p$-адическая квазимера Гиббса.\\
iii) Пусть $p>3$ и $x_0$ является решением сравнения $x^2+3\equiv0
\,(\operatorname{mod} p)$. Тогда существуют две 2-периодические
$p$-адические квазимеры Гиббса в том и только в том случае, если
сравнение $x^2-2x_0+2\equiv0 \,(\operatorname{mod} p)$ разрешимо в
$\mathbb Q_p$.
\end{thm}
\begin{proof}
Так как существования $\sqrt{D(\theta)}$ и $\sqrt{-3}$
эквивалентны, то мы можем рассмотреть только случай, когда
сравнение $x^2+3\equiv0 \,(\operatorname{mod} p)$ разрешимо в
$\mathbb Q_p$. Заметим, что $\sqrt{-3}\not\in\mathbb Q_p$ при
$p\leq3$.

Пусть $p>3$ и $x_0$ является решением сравнения $x^2+3\equiv0
\,(\operatorname{mod} p)$. Тогда
$$u_3=\frac{-1-p^{4J}+\sqrt{1+2p^{4J}-3p^{8J}}}{2p^{4J}}=\frac{x_0-1+p^{-4J}\varepsilon}{2},\quad |\varepsilon|_p\leq1$$
Отсюда из Теоремы \ref{tx2} следует, что существование
$\sqrt{u_3}$ эквивалентно разрешимости сравнения
$x^2-2x_0+2\equiv0 \,(\operatorname{mod} p)$.

В силу замечания  \ref{perrk} существуют две 2-периодические
$p$-адические квазимеры Гиббса в том и только в том случае, если
сравнение $x^2-2x_0+2\equiv0 \,(\operatorname{mod} p)$ разрешимо в
$\mathbb Q_p$.
\end{proof}

{\bf Благодарность.} Автор благодарен У.А. Розикову и Ф.М.Мухамедову за полезные
советы и ряд замечаний.

\end{document}